\documentclass[12pt]{article}
\usepackage{geometry}
\usepackage{amssymb}
\usepackage{amstext,enumerate}
\usepackage{bm}
\usepackage{graphics}
\usepackage{graphicx}
\usepackage{amsthm}
\usepackage{amsmath}
\usepackage{latexsym}
\usepackage{rotate}
\usepackage{lscape}
\usepackage{color}
\usepackage{epsfig,epsf,psfrag}
\usepackage{sectsty}
\usepackage{graphics}
\usepackage{epsfig}
\usepackage{multirow}
\usepackage{float}
\usepackage{graphicx}
\usepackage{amsmath}
\usepackage{url}
\usepackage{amsthm}
\usepackage{amssymb}
\usepackage{sectsty}
\usepackage{epsfig,natbib}
\usepackage{rotate}
\usepackage{lscape}
\usepackage{setspace}
\usepackage{subcaption}
\usepackage{float}
\usepackage{graphicx}
\usepackage{mathtools}
\usepackage{amstext,amssymb,amsfonts,amscd,color}
\usepackage{graphicx}
\usepackage{epsfig, natbib}
\usepackage{threeparttable}
\usepackage{cases,color}
\usepackage{multirow}
\usepackage{verbatim}
\usepackage{booktabs}
\usepackage{multicol}
\usepackage{bbm}
\usepackage{bm}
\usepackage{subcaption}
\usepackage[ruled,vlined]{algorithm2e}
\usepackage{comment}
\usepackage{multirow}

\include{mathdefn2}
\usepackage{mathrsfs}
\usepackage{float}

\usepackage{xparse}
\usepackage{tikz}
\usetikzlibrary{positioning}
\usetikzlibrary{plotmarks}
\usetikzlibrary{matrix}
\usepackage{pgfplots}
\pgfplotsset{compat=1.7}
\usetikzlibrary{matrix,backgrounds, arrows.meta}
\pgfdeclarelayer{myback}
\pgfsetlayers{myback,background,main}
\usetikzlibrary{decorations.pathreplacing,angles,quotes}
\tikzset{mycolor/.style = {dashed,rounded corners,line width=1bp,color=#1}}%
\tikzset{myfillcolor/.style = {draw,fill=#1}}%
\tikzset{
	declare function={
		normcdf(\x,\m,\s)=1/(1 + exp(-0.07056*((\x-\m)/\s)^3 - 1.5976*(\x-\m)/\s));
	}
}

\sectionfont{\centering\bf\sf\normalsize}
\subsectionfont{\sf\normalsize}

\renewcommand{\baselinestretch}{1.6} 
\newcommand{\single}{\renewcommand{\baselinestretch}{1.2}\normalsize}
\newcommand{\double}{\renewcommand{\baselinestretch}{1.63}\normalsize}

  \renewenvironment{thebibliography}[1]{
    \begin{oldthebibliography}{#1}
      \setlength{\parskip}{0ex}
      \setlength{\itemsep}{0ex}
  }
  {
    \end{oldthebibliography}
  }

\newcommand{\bea}{\begin{eqnarray*}}
\newcommand{\eea}{\end{eqnarray*}}
\newcommand{\be}{\begin{eqnarray}}
\newcommand{\ee}{\end{eqnarray}}
\newcommand{\ed}{\end{document}}

\newcommand{\btab}{\begin{tabular}}
\newcommand{\etab}{\end{tabular}}

\newcommand{\bi}{\begin{itemize}}
\newcommand{\ei}{\end{itemize}}
\newcommand{\bfi}{\begin{figure}}
\newcommand{\efi}{\end{figure}}
\newcommand{\ben}{\begin{enumerate}}
\newcommand{\een}{\end{enumerate}}
\newcommand{\bay}{\begin{array}}
\newcommand{\eay}{\end{array}}

\definecolor{DarkBlue}{rgb}{0,.08,.45}
\definecolor{DarkRed}{rgb}{.7,0,.4}

\def\hg #1 {\texcolor{cyan}{{\it Hans:}   #1}}

\def\bco{\iffalse}

\def\var{{\rm var}}

\newcommand{\no}{\noindent}
\newcommand{\bc}{\begin{center}}
\newcommand{\ec}{\end{center}}

\newcommand{\bsp}{\begin{split}}
\newcommand{\esp}{\end{split}}
\newcommand{\bdes}{\begin{description}}
\newcommand{\edes}{\end{description}}
\newcommand{\bass}{\begin{assumption}}
\newcommand{\eass}{\end{assumption}}
\newcommand{\bthm}{\begin{theorem}}
\newcommand{\ethm}{\end{theorem}}
\newcommand{\blem}{\begin{lemma}}
\newcommand{\elem}{\end{lemma}}

\def\bco{\iffalse}

\def\var{{\rm var}}

\bibpunct{(}{)}{;}{a}{}{,}

\newtheorem{assumption}{Assumption}
\newtheorem{theorem}{Theorem}
\newtheorem{corollary}{Corollary}
\newtheorem{example}{Example}
\newtheorem{remark}{Remark}
\newtheorem{lemma}{Lemma}
\newtheorem*{definition}{Definition of Marginal Homogeneity}

\newcommand{\JL}[1]{{\color{cyan}\bf [JL: #1]}}

\begin{document}

\def\spacingset#1{\renewcommand{\baselinestretch}%
{#1}\small\normalsize} \spacingset{1}

\thispagestyle{empty} \single \bc {\bf \sc \Large Testing Homogeneity: The Trouble with Sparse Functional Data}
\vspace{0.15in}\\
Changbo Zhu and  Jane-Ling Wang \\
Department of Statistics, University of California, Davis \\
Davis, CA 95616 USA \ec \centerline{2 July 2022}

\vspace{0.1in} \thispagestyle{empty}
\bc{\bf \sf ABSTRACT} \ec \vspace{-.1in} \no 
\setstretch{1} 
Testing the homogeneity between two samples of functional data is an important task. While this is feasible for intensely measured functional data, we explain why  it is challenging for sparsely measured functional data and show what can be done for such data. In particular, we show that testing the marginal homogeneity based on point-wise distributions is feasible under some constraints and propose a new two sample statistic that works well with both intensively and sparsely measured functional data. The proposed test statistic  is formulated upon Energy distance, and the critical value is obtained via the permutation test. The convergence rate of the test statistic to its population version is derived along with the consistency of the associated permutation test. To the best of our knowledge, this is the first paper that provides guaranteed consistency for testing the homogeneity for sparse functional data. The aptness of our method is demonstrated on both synthetic and real data sets. \\   

\no {KEY WORDS:\quad Longitudinal data; Sparse functional data; Two sample test; Energy distance; Measurement errors; Convergence rate}.
\thispagestyle{empty} \vfill
\noindent \vspace{-.2cm}\rule{\textwidth}{0.5pt}\\
{\small Changbo Zhu acknowledges the support of NIH Echo UH3OD023313. The research of Jane-Ling Wang is supported by NSF grant DMS-1914917.} 

\newpage
\pagenumbering{arabic} \setcounter{page}{1} \double

\section{Introduction}
Two sample testing of equality of  distributions, which is the homogeneity hypothesis, is fundamental in statistics and has a long history that dates back to \cite{kolmogorov1933,smirnov1948,cramer1928,von1928}. The literature on this topic can be categorized by the data type considered. Classical tests designed for low dimensional data include \cite{bickel1969, bickel1983, friedman1979, henze1988, schilling1986} among others. For recent developments that are applicable to data of arbitrary dimension, we refer to Energy Distance (ED) \citep{szekely2004} and Maximum Mean Discrepancy (MMD) \citep{sejdinovic2013nips}. To suit high dimensional regimes \citep{hall2005, aoshima2018survey},  extensions of Energy Distance and Maximum Mean Discrepancy were studied in \cite{zhu2019interpoint, chakraborty2019new, gao2021two}. Some other interesting developments of Energy distance for data residing in a metric space include \cite{lyons2013} and \cite{klebanov2005n}.
In this paper, we focus on  functional data that are random samples of functions on a real interval, e.g. $[0, 1] $ 
\citep{ramsay2004functional, hsing2015theoretical, 
doi:10.1146/annurev-statistics-041715-033624, davidian2004introduction}.

Two sample inference for functional data is gaining more attention due to the explosion of data that can be represented as functions. A substantial literature  has been devoted to  comparing the mean and covariance functions between two groups of curves, see \cite{fan1998test, cuevas2004anova,10.2307/20441490, doi:10.1080/15598608.2010.10412005, doi:10.1198/jasa.2010.tm09239, horvath2012inference,zhang2014one, STAICU20151, pini2016interval, 10.1093/biomet/asw033,     zhang2019new,    guo2019new, yuan2020hypothesis, 10.1214/21-EJS1802}. However, the underlying distributions of two random functions can have the same mean and covariance function, but differ in other aspects. Testing homogeneity, which refers to a hypothesis testing procedure to determine the equality of the underlying distributions of any two random objects, is thus of particular interest and practical importance.
The literature on testing the homogeneity for functional data is much smaller and restricted to the two-sample case based on either fully observed \citep{cabana2017permutation, benko2009common, wynne2020kernel,krzysko2021two} or intensely measured functional data \citep{hall2007two, jiang2019asymptotics}.

For intensively measured functional data, a presmoothing step is often adopted to  each individual curve in order to construct a smooth curve before carrying out subsequent analysis, which may  reduce mean square error \citep{10.2307/24310140} or, with luck, remove the noise, a.k.a. measurement error, in the observed discrete data. 
 This results in a two-stage procedure, smoothing first and then testing homogeneity based on the presmoothed curves. For example, \cite{hall2007two, jiang2019asymptotics} adopted such an approach. 
In addition, the energy distance in \citep{szekely2004} could be extended to the space of $L^2$ functions to characterize the distribution differences of random functions,  provided the functional data are  fully observed without errors \citep{klebanov2005n}. 

Specifically, the energy distance between two random functions $X$ and $Y$  is defined as
\begin{align}
\label{eq:ed}
\text{ED}(X,Y) = 2E[\| X-Y \|_{L^2}] - E[\| X-X' \|_{L^2}] - E[\| Y-Y' \|_{L^2}],
\end{align}
where $X',Y'$ are i.i.d copies of $X,Y$ respectively and $ \| X-Y \|_{L^2} = (\int_{0}^1 (X(t) - Y(t))^2 dt)^{1/2} $. According to the results of  \cite{lyons2013} and \cite{klebanov2005n}, $ \text{ED}(X,Y) $  fully characterizes the distributions of $X$ and $Y$ in the sense that $ \text{ED}(X,Y) \geq 0 $ and $ \text{ED}(X,Y)=0 \Leftrightarrow X =^d Y $, where  we use $X=^d Y$ to indicate that $X,Y$ are identically distributed. Given two samples of functional data $\{ X_i \}_{i=1}^n$ and $\{Y_i\}_{i=n+1}^{n+m},$ which are intensively measured at some discrete time points, the reconstructed functions, denoted as $\{ \widehat{X}_i \}_{i=1}^n$ and $\{ \widehat{Y}_i\}_{i=n+1}^{n+m}$, can be obtained using the aforementioned presmoothing procedure. Then $\text{ED}(X,Y)$ can be estimated by the $U$-type statistic
\begin{multline} \label{eq:2step}
\text{ED}_n (X,Y) =  \frac{2}{mn} \sum\limits_{i_1=1}^{n} \sum\limits_{i_2=n+1}^{n+m} \| \widehat{X}_{i_1}-\widehat{Y}_{i_2} \|_{L^2} \\ - \frac{2}{n(n-1)}  \sum\limits_{1 \leq i_1 < i_2 \leq n} \| \widehat{X}_{i_1}-\widehat{X}_{i_2} \|_{L^2} -  \frac{2}{m(m-1)} \sum\limits_{n+1 \leq i_1 < i_2 \leq n+ m} \| \widehat{Y}_{i_1}-\widehat{Y}_{i_2} \|_{L^2}.
\end{multline}

While the above presmoothing procedure to reconstruct the original curves  may be promising for intensely measured functional data, it has not yet been utilized  to our knowledge,  perhaps  due to the technical and practical challenges to implement this approach as the level of intensity in the measurement schedule and the proper amount of smoothing are both critical.  First, the distance between the reconstructed functional data and its target (the true curve) needs to be tracked and reflected in the subsequent calculations.  Second, such a  distance would depend on the intensity of the measurements and  the bandwidth used in the presmoothing stage. Neither is easy to nail down in practice.  

Furthermore, in real world applications, such as in longitudinal studies, each subject often may  only have a few measurements, leading to sparse functional data \citep{yao2005functional}. Here presmoothing individual data no longer works and one must borrow information from all subjects to reconstruct  the trajectory of an individual subject. 
The PACE approach in \citet{yao2005functional} offers such an imputation method, yet it does not lead to consistent estimates of  the true curve  for sparsely observed functional data as  there are  not enough data available for each individual. 
Consequently, the quantities $E[\| X-Y \|_{L^2}], E[\| X-X' \|_{L^2}]\  \text{and} \  E[\| Y-Y' \|_{L^2}]$ in equation \eqref{eq:ed} are not consistently estimable as these expectations are outside the corresponding $L^2$ norms. 
\cite{10.2307/24773028} reduced the problem to testing the homogeneity of the scores of the two processes by assuming that the random functions are finite dimensional. Such an approach would not be consistent either as the scores still cannot be consistently estimated for sparse functional data. To our knowledge, there exists no consistent  test of homogeneity for sparse functional data.  In fact, it is not feasible  to test full homogeneity based on sparsely observed functional data as there are simply not enough data for such an ambitious goal.

This seems disappointing, since although it has long been recognized that sparse functional data are much more challenging to handle than intensively measured functional data, much progress has been made to resolve this challenge. For instance, both the  mean and covariance function can be estimated 
consistently at a certain rate \citep{yao2005functional, li2010uniform, zhang2016sparse} for sparsely observed functional data. Moreover, the regression coefficient function in a functional linear model can also be estimated consistently with rates  \citep{yao2005regression}. This motivated us to explore a less stringent concept of homogeneity that can be tested consistently for sparse functional data. 
In this paper we provide the answer by proposing a  test of marginal homogeneity for two independent samples of functional data. For ease of presentation we assume that the random functions are defined on the unit interval $[0, 1]$. 
\begin{definition} 
Two random functions $X$ and $Y$ defined on [0, 1] 
are  marginal homogeneous if $$\ \; X(t) =^d Y(t) \text{ for almost all } t \in [0,1] .$$
\end{definition}

From this definition we can see that, unlike testing   homogeneity that involves testing the entire distribution of functional data, testing marginal homogeneity only involves simultaneously testing  the marginal distributions at all time points.   This is a much more manageable task that works for all sampling designs, be it intensively or sparsely observed functional data, and it is often adequate in many applications. Testing  marginal homogeneity is not new in the literature and has been investigated by \cite{zhu2019interpoint, chakraborty2019new} for high-dimensional data. They show that the marginal tests can be more powerful than their joint counterparts under the high dimensional regime. In a larger context, the idea of aggregating marginal information originates from \cite{zhu2020distance}, where they consider a related problem of testing the independence between two high-dimensional random vectors. 

For real applications of testing marginal homogeneity, taking the analysis of biomarkers over time in clinical research as an example, comparing differences between marginal distributions of the treatment and control groups 
may be sufficient  to establish the  treatment effect. 
To contrast stocks in two different sectors, the differences between marginal distributions might be more important than the differences between joint distributions. In addition, differences between marginal distributions can be seen as the main effect of differences between  distributions. Thus, it makes good sense to test marginal homogeneity, especially in situations where joint distribution testing is not feasible or inefficient. 

Let $\lambda$ be the Lebesgue measure on $\mathbb{R}$. The focus of this paper is  to test
\begin{align} \label{hypotheses}
\begin{array}{c}
H_0: X(t) =^d Y(t) \text{ for almost all } t \in [0,1], \\
\text{versus} \\
H_A : \text{there exists a set } \mathbb{T} \subseteq [0,1] \text{ such that } \lambda(\mathbb{T}) >0 \text{ and } X(t) \neq^d Y(t) \text{ if } t \in \mathbb{T}.
\end{array}
\end{align}

This can be accomplished through the  marginal energy distance (MED) defined as:
\begin{align} \label{MED}
\text{MED}(X,Y)=\int 2E \left[ \left| X(t) - Y(t) \right| \right]  -  E \left[   \left| X(t) - X'(t) \right|  \right] - E \left[   \left| Y(t) - Y'(t) \right| \right]dt ,
\end{align}
where $X'$ and $ Y'$ are independent copies of $X$ and $Y$ respectively. Indeed, MED is a metric for marginal distributions in the sense that $ \text{MED}(X,Y) \geq 0 $ and $\text{MED}(X,Y)=0 \Leftrightarrow X(t) = ^d Y(t)$ for almost all $t \in [0,1]$. 
A key feature of our approach is that it can consistently estimate  $\text{MED}(X,Y)$ for all types of sampling plans. Moreover,  $ E \left[ \left| X(t) - Y(t) \right| \right] $, $ E \left[   \left| X(t) - X'(t) \right|  \right] $ and $ E \left[   \left| Y(t) - Y'(t) \right| \right] $ can all be reconstructed consistently for both intensively and sparsely observed functional data. Such a unified procedure for all kinds of sampling schemes may be more practical as the separation between intensively and sparsely observed functional data is usually unclear in practical applications. Moreover, it could happen that while some of   the subjects are intensively observed, others are sparsely observed. In the extremely sparse case,  our method can still work if each subject only has one measurement.


Measurement errors (or noise) are common for functional data, so it is important to accommodate them. If noise is left unattended,  there will be 
bias in the estimates of MED 
as the observed distributions are no longer the true distributions of $X$ and $Y$. One might hope that the measurement errors can be averaged out during the estimation of the function $E[ | X(t) - Y(t) |]$ in $\text{MED}(X,Y)$ 
in analogy to  estimating the mean function $\mu(t) = E[X(t)]$ or covariance function $C(s,t) = E[(X(t) - \mu(t))(X(s)  -\mu(s))] $  \citep{yao2005functional}. However,   this is not the case. To see why, let $e_1$, $e_1'$, $e_2$ and $e_2'$ be independent white noise. When estimating mean or covariance function at any fixed time $t, s \in [0,1] $, it holds that $\mu(t) =  E[X(t) + e_1] $ and $C(s,t) = E [(X(t) - \mu(t) + e_1) ( X(s) - \mu(s) + e_1' )] $ for $t \neq s$. 
But $ E[| X(t) - Y(s) + e_1 - e_2 |] \neq E[ | X(t) - Y(s) |] $. Likewise, we can see that the energy distance ED in \eqref{eq:ed} would have  the same challenge to handle measurement errors unless these errors  were removed in a presmoothing step before carrying out the test.  So the challenges  with measurement errors  is not triggered by the use of the $L^1$ norm in MED. The $L^2$ norm in ED will  face the same challenge.

For intensely measured functional data, a presmoothing step is often used to handle  measurement errors in the observed data in the hope that smoothing will remove the error.   However, this is a delicate issue, as   
it is difficult to know the amount of smoothing needed in order for the subsequent analysis to retain the same convergence rate as if the true functional data were fully observed without errors. For instance, \cite{10.1214/009053606000001505} study the effects of smoothing to obtained reconstructed curves and show that in order to retain the same convergence rate of mean estimation for fully observed functional data  the number of measurements per subject that generates the curves must be of higher order than the number of independent subjects. This requires functional data that are intensively sampled well beyond ultra dense (or dense) functional data that have been studied  in the literature 
\citep{zhang2016sparse}. 

In this paper, we propose a new way of handling measurement errors so that the $\text{MED}$-based testing procedure is still consistent in  the presence of measurement errors. The key idea is to show that when the measurement errors $e_1$ and $ e_2$ of $X$ and $Y$, respectively, are identically distributed, i.e, $e_1 = ^d e_2$, the $\text{MED}(X,Y)$-based approach can still be applied to the contaminated data with consistency guaranteed under mild assumptions (cf. Corollary \ref{cor:powerconta}). 
When $e_1 \neq^d e_2$, we 
propose an error-augmentation approach, which can be applied jointly with our unified estimation procedure. 

The rest of the paper is organized as follows.  Section \ref{sec:3} contains the main methodology and supporting theory about testing marginal homogeneity. Numerical studies are presented in Section \ref{sec:sim} The conclusion is in Section \ref{sec:con} All technical details are postponed to Section \ref{sec:tec}

\section{Testing Marginal Homogeneity} \label{sec:3}

We first consider the case where there are no measurement errors and postpone the discussion of measurement errors to the end of this section. Let $\{ X_i \}_{i=1}^n$ and $\{Y_i\}_{i=n+1}^{n+m}$ be  i.i.d copies of $X$ and $Y$,  respectively. In practice, the functions are only observed at some discrete points, i.e.,
\begin{align*}
\begin{array}{ll}
x_{ij} = X_i(T_{ij}), & \text{ if } i = 1,2, \dots, n, \ j = 1,2, \dots, N_i , \\
y_{ij} = Y_{i}(T_{ij}), & \text{ if } i = n+1, \dots,n+ m, \  j = 1,2, \dots, N_i.
\end{array}
\end{align*}
This sampling plan allows the two samples  to be measured at different schedules and additionally each subject within the sample  has its own measurement schedule. This is a realistic assumption but the consequence is that two-dimensional smoothers will be needed to estimate the targets. Fortunately,  the convergence rate of our estimator attains the same convergence rate as that of a one-dimensional  smoothing method. This intriguing phenomenon will be explained later.

For notational convenience, denote $ Z_i = X_i,  \text{ if } i = 1,2, \dots, n$ and $ Z_i = Y_i,  \text{ if } i =n+1, \dots, n+m $. Let $\mathbf{Z} = (\mathbf{z}_1 ; \dots ;  \mathbf{z}_{n+m}  )$ be the combined observations, where for $i=1,2, \dots, n+m$, $\mathbf{z}_i $ is a vector of length $N_i$,
\begin{align*}
\mathbf{z}_i= (z_{i1}, z_{i2}, \dots, z_{i,N_i})^T = \left\lbrace
    \begin{array}{ll}
    (x_{i1}, x_{i2}, \dots , x_{i,N_i})^T,   & \text{ if } 1 \leq i \leq n,  \\
    (y_{i1}, y_{i2}, \dots , y_{i,N_i})^T,   & \text{ if } n+ 1 \leq i \leq n +m.
    \end{array} \right.
\end{align*}
The observations corresponding to $X$ and $Y$ are defined as  $\mathbf{X} = (\mathbf{z}_1; \dots ; \mathbf{z}_n)$ and $ \mathbf{Y} = (\mathbf{z}_{n+1}; \dots ; \mathbf{z}_{n+m})$ respectively. 

To estimate $\text{MED}(X,Y)$ in (\ref{MED}), note that we actually have no observations for the one-dimensional functions $ E \left[ \left| X(t) - Y(t) \right| \right] $, $E \left[ \left| X(t) - X'(t) \right| \right]$ and $ E \left[ \left| Y(t) - Y'(t) \right| \right] $,  due to the longitudinal design where $X$ and $Y$ are observed at different time points. Thus, the sampling schedule for $X$ and $Y$ are not synchronized.  A consequence of such asynchronized functional/longitudinal data is that a one-dimensional smoothing method that has  typically been employed to estimate a one-dimensional target function, e.g. $ E \left[ \left| X(t) - Y(t) \right| \right], $ does not work here. However, a workaround is to  estimate the following two-dimensional functions first: 
\begin{align*}
&G_{1}(t_1,t_2) := E\left[  \left| X(t_1) - Y(t_2) \right| \right], \\
&G_{2} (t_1,t_2):= E \left[   \left| X(t_1) - X'(t_2) \right|  \right], \\
&G_{3}(t_1,t_2):= E \left[   \left| Y(t_1) - Y'(t_2) \right| \right], 
\end{align*}
then set $t_1=t_2=t $ in all three estimators. 
Since $G_1, G_2, G_3$ can all be recovered by some local linear smoother, $\text{MED}(X,Y)$ admits consistent estimates for both intensively and sparsely observed functional data. For instance, $G_{1}(t_1,t_2)$ can be estimated by $\widehat{G}_{1}(t_1,t_2) = \hat{\beta}_0$, where
\begin{multline} \label{eq:xy}
(\hat{\beta}_0, \hat{\beta}_1, \hat{\beta}_2) =  \underset{\beta_0, \beta_1, \beta_2}{\text{argmin}} \frac{1}{nm} \sum_{1 \leq i_1 \leq n} \sum_{n+1 \leq i_2  \leq n+ m} \frac{1}{N_{i_1}} \frac{1}{N_{i_2}} \sum_{j_1=1}^{N_{i_1}} \sum_{j_2=1}^{N_{i_2}} K_{h_x}(T_{i_1 j_1} - t_1) \times \\   K_{h_y}(T_{i_2 j_2}- t_2)   \left[\left| z_{i_1 j_1} - z_{i_2 j_2} \right|  - \beta_0 - \beta_1 (T_{i_1j_1}-t_1)  - \beta_2 (T_{i_2j_2} - t_2)  \right]^2,
\end{multline}
and $G_2(t_1, t_2)$ can be estimated by $\widehat{G}_{2}(t_1, t_2) = \hat{\alpha}_0$, where
\begin{multline} \label{eq:xx}
(\hat{\alpha}_0, \hat{\alpha}_1, \hat{\alpha}_2) = \underset{\alpha_0, \alpha_1, \alpha_2}{\text{argmin}} \frac{2}{n(n-1)} \sum_{1 \leq i_1 < i_2 \leq n}  \frac{1}{N_{i_1}} \frac{1}{N_{i_2}}  \sum_{j_1=1}^{N_{i_1}} \sum_{j_2=1}^{N_{i_2}}  K_{h_x}(T_{i_1j_1}- t_1)   K_{h_x}(T_{i_2j_2}- t_2)  \\  [ \left| z_{i_1j_1} - z_{i_2 j_2} \right| - \alpha_0 - \alpha_1 (T_{i_1j_1}-t_1)  - \alpha_2 (T_{i_2 j_2} - t_2) ]^2.
\end{multline}
$G_{3}(t_1,t_2)$ can be estimated similarly as $G_{2}(t_1,t_2)$ by an estimator $\widehat{G}_{3}(t_1,t_2)$. In the above, $N_{i1}$ and $N_{i2}$ should be understood as the respective length of the vector $\mathbf{z}_{i1}$ and $\mathbf{z}_{i2}$  
and  $K_{h}(\cdot) = K(\cdot/h)/h$ is a one-dimensional kernel with bandwidth $h$.
The sample estimate of $ \text{MED}(X,Y) $ can be constructed as
\begin{align} \label{eq:MED}
\text{MED}_n(\mathbf{Z}) := \int_{0}^1 2 \widehat{G}_{1}(t,t) - \widehat{G}_{2}(t, t) - \widehat{G}_{3}(t, t)  dt.
\end{align}



For hypothesis testing, the critical value or $p$-value can be determined by permutations \citep{lehmann2006testing}. To be more specific, let $\pi : \{ 1,2, \dots, n+m \} \rightarrow  \{ 1,2, \dots, n+m \}$ be a permutation. There are $(n+m)!$ number of permutations in total and we denote the set of permutations as $\mathbb{P}_{n+m} = \{ \pi_l : l=1,2, \dots, (n+m)! \}$. For $l=1,2, \dots, (n+m)!$, define the permutation  of $\pi_l$ on $\mathbf{Z}$ as:
\begin{align}
    \pi_l \cdot \mathbf{Z} = (\mathbf{z}_{\pi_l(1) } ; \mathbf{z}_{\pi_l(2) }; \dots ;  \mathbf{z}_{\pi_l(n+m)}  ).
\end{align}
Write the statistic that is based on the permuted sample $\pi_l \cdot \mathbf{Z} $ as $ \text{MED}_n(\pi_l \cdot \mathbf{Z}) $ and let  $ \Pi_1, \dots, \Pi_{S-1} $ be i.i.d and uniformly sampled from $ \mathbb{P}_{n+m} $, we define the permutation based $p$-value as 
\begin{align*}
    \widehat{p} = \frac{1}{S} \left\{ 1 + \sum\limits_{l=1}^{S-1}  \mathbb{I}_{ \left\lbrace \text{MED}_n( \Pi_l \cdot \mathbf{Z}) \geq \text{MED}_n(  \mathbf{Z})  \right\rbrace } \right\}.
\end{align*}
Then, the level-$\alpha$ permutation test w.r.t. $\text{MED}_n( \mathbf{Z})$ can be defined as:
	$$
	\text{Reject } H_0,  \text{ if } \widehat{p} \leq \alpha. 	
	$$

\subsection{Convergence Theory}
In this subsection, we show that $\text{MED}_{n}(\mathbf{Z})$ is a consistent estimator and develop its convergence rate. 
\begin{assumption} \label{ass:1} 
	\begin{itemize}
		\item[A.1] The kernel function $K(\cdot) \geq 0$ is symmetric, Lipschitz continuous, supported on $[-1, 1]$ and satisfies 
		\begin{align*}
		\int K(u)du = 1, \;  \int_{0}^1 u^2 K(u) du < \infty \text{ and } \int_{0}^1 K(u)^2 du < \infty.
		\end{align*}
		\item[A.2] Let $\{ T_{ij} :  1 \leq i \leq n, 1 \leq j \leq N_i \} \sim^{i.i.d} T_x$, $\{ T_{ij}  :  n+ 1 \leq i \leq n +m, 1 \leq j \leq N_i \} \sim^{i.i.d} T_y$ and denote  the density functions of $T_x, T_y$ by  $g_{x}$, $g_y$ respectively. There exists constants $c$ and $ C$ such that $0 < c \leq g_{x} (s), g_{y} (t) \leq C < \infty$ for any $s,t \in [0, 1]$.
		\item[A.3] $\{ X_{i_1}, Y_{i_2}, T_{ij} : 1 \leq i_1 \leq n, n+1 \leq i_2 \leq n+m, 1 \leq i \leq n+m, 1 \leq j \leq N_i \}$ are mutually independent.
		\item[A.4]  The second order partial derivatives of $G_{1}, G_{2}, G_{3}$ are bounded on $[0,1]$.
		\item[A.5]  $ \sup_{t} E|X(t)|^2 < \infty $ and $\sup_{t} E|Y(t)|^2 < \infty  $.
	\end{itemize}
\end{assumption}

\begin{remark} \label{eq:rmk}
	Conditions A.1 - A.3 and A.5 are fairly standard and also used in \cite{li2010uniform}.
	Condition A.4 may seem  a bit problematic at first, as the absolute value function $| \cdot | $ is not differentiable at 0. However, its expectation can easily be differentiable. For instance, if  the density functions of $X(t)$, $Y(t)$ are  $f_{x}(\cdot|t)$, $f_{y}(\cdot|t)$ respectively, then we have $G_{1}(s,t) = \int\int |u-v| f_{x}(u|s)f_{y}(v|t) dudv$. Assuming the conditions of the  Leibniz integral rule, we can interchange the partial derivatives and integration, i.e.,
	\begin{align*}
	    \frac{\partial^2 }{\partial s \partial t} G_{1}(s,t) = \int\int |u-v|  \frac{\partial }{\partial s} f_{x}(u|s) \frac{\partial }{\partial t} f_{y}(v|t) dudv.
	\end{align*}
	Thus, the partial derivatives of $G_{1}(s,t)$ are bounded if the second order partial derivatives of $ f_{x}(u|s) $, $ f_{y}(v|t) $ w.r.t. $s,t$ exist for all $u,v$ and 
	\begin{align}\label{eq:partial}
    \sup_{u} \left| \frac{\partial^2 }{ \partial s \partial s }  f_{x}(u|s)  \right| < \infty \text{ and } \sup_{v} \left| \frac{\partial^2 }{ \partial t \partial t }  f_{y}(v|t) \right| < \infty.
    \end{align}
	A more  specific example is when  $X(t)$ is Gaussian and $Y(t)$ is a mixture of Gaussians with density functions 
	\begin{align*} 
	f_{x}(u|t) & = \frac{1}{\sigma_{1}(t) \sqrt{2\pi}} e^{ - \frac{1}{2} \left( \frac{u - \mu_{1}(t)}{\sigma_{1}(t) } \right)^2 }, \\
	f_{y}(u|t) & = \frac{1}{2} \frac{1}{\sigma_{2}(t) \sqrt{2\pi}} e^{ - \frac{1}{2} \left( \frac{u - \mu_{2}(t)}{\sigma_{2}(t) } \right)^2 } + \frac{1}{2} \frac{1}{\sigma_{2}(t) \sqrt{2\pi}} e^{ - \frac{1}{2} \left( \frac{u + \mu_{2}(t)}{\sigma_{2}(t) } \right)^2 }.
	\end{align*}
	Then A.4 holds if we assume $\sigma_{1}(s), \sigma_{2}(t) $ are bounded from below by a positive constant, $ \mu_1(s), \mu_2(t), \sigma_{1}(s), \sigma_{2}(t)$ are bounded and have bounded second order derivatives. Similar conclusions can be drawn for $G_2$ and $G_3$. Therefore, Condition A.4 is not restrictive as it is customary to assume that the mean and covariance functions for functional data are differentiable. 

\end{remark}
The next assumption specifies the relationship of the number of observations per subject and the decay rate of the  bandwidth parameters $h_x, h_y$. 
\begin{assumption} \label{ass:2new}
 Suppose  $h_x:=h_x(n), h_y:=h_y(m) \rightarrow 0$ and
\begin{align*}
 & \log \left( \frac{n}{ \sum_{i=1}^n N_i^{-1}/n} \right) \frac{\max_{1 \leq i \leq n} N_i^{-1}}{h_x} \frac{\max_{1 \leq i \leq n} N_i^{-1}}{ \sum_{i=1}^n N_i^{-1}/n} \frac{1}{ n}  \rightarrow 0, \\
  & \log \left( \frac{m}{ \sum_{i=n+1}^{n+m} N_i^{-1}/m} \right) \frac{\max_{n+1 \leq i \leq n+m} N_i^{-1}}{h_y} \frac{\max_{n+1 \leq i \leq n+m} N_i^{-1}}{ \sum_{i=n+1}^{n+m} N_i^{-1}/m} \frac{1}{m}  \rightarrow 0.
\end{align*}
\end{assumption}

The following theorem states that we can consistently estimate $\text{MED}(X,Y)$ with sparse observations.

\begin{theorem} \label{thm:1} Under Assumptions \ref{ass:1} and \ref{ass:2new},
	\begin{align*}
	\left| \emph{\text{MED}}_n(\mathbf{Z}) - \emph{\text{MED}}(X,Y) \right| = O_p\left(h_x^2+ \sqrt{ \frac{1}{n^2}  \sum_{i=1}^{n} \phi_i }   + h_y^2 + \sqrt{ \frac{1}{m^2}  \sum_{i=n+1}^{n+m} \phi_i } \right),
	\end{align*}
	where $\{\phi_i: i=1,2, \dots, n+m\}$ are defined as 
	\begin{align*}
	    \phi_i = \left\lbrace
	    \begin{array}{ll}
	      \frac{N_ih_x + N_i(N_i-1)h_x^2}{ N_i^2 h_x^2},   &  1\leq i \leq n , \\
	       \frac{N_i h_y + N_i(N_i-1)h_y^2}{N_i^2 h_y^2},  & n+1 \leq i \leq n+m .
	    \end{array} \right.
	\end{align*}
\end{theorem}

\begin{remark}
For any two-dimensional function $F \in L^2([0,1]^2)$, define the $L^2$-norm as $\|F\|_2 := (\int_{t_1}\int_{t_2} [F(t_1, t_2)]^2 dt_1 dt_2  )^{1/2}$. The above theorem is a consequence of 
\begin{align*}
    \left\| G_I(t_1, t_2) - \widehat{G}_{I}(t_1, t_1) \right\|_2 & = O_p\left(h_x^2+ \sqrt{ \frac{1}{n^2}  \sum_{i=1}^{n} \phi_i }   + h_y^2 + \sqrt{ \frac{1}{m^2}  \sum_{i=n+1}^{n+m} \phi_i } \right),
\end{align*}
where $I=1,2,3$. Compared with the mean function $\mu(t) = E[X(t)]$ and the covariance function $ C_X(s,t) = E[(X(s) - \mu(s))(X(t) - \mu(t))] $, $ G_1, G_2, G_3$ are functions involving two independent stochastic processes. 
An intriguing  phenomenon is that even though $ G_1, G_2, G_3 $ are two-dimensional functions, the convergence rate of their linear smooth estimates is the same as for a one-dimensional function, such as the mean function. 
This is because the expectation $G_1(s,t) = E[|X(s) - Y(t)|]$ involves two independent stochastic processes and  $n \times m$ pairs $\{ (X_{i_1}, Y_{i_2}):i_1=1,\dots, n, i_2=n+1, \dots,n+m  \}$ are used in the linear smoother, leading to a faster convergence rate.  This distinguishes this situation from  the  standard estimation of a bivariate function. For instance, if the goal is to estimate  $E[|X(s) - X(t)|]$, the convergence rate would be slower and would be the same as that for a two-dimensional smoother. We further point out that even though a two dimensional smoothing method is used to estimate $G_I (s, t)$, we only need to evaluate its values at the diagonal where $s=t$. Therefore, the computational effort is manageable. 
\end{remark}

Given two sequences of positive real numbers $a_n$ and $b_n$, we say that $a_n$ and $b_n$ are of the same order as $n \rightarrow \infty$ (denoted as $a_n \asymp b_n$) if there exists constants $0 < c_1 < c_2 < \infty$ such that $ c_1 \leq \lim_{n \rightarrow \infty} a_n/b_n \leq c_2 $ and  $ c_1 \leq \lim_{n \rightarrow \infty} b_n/a_n \leq c_2 $. The convergence rates of $\text{MED}_n(\mathbf{Z})$ for different sampling plans are provided in the following corollary. 


\begin{corollary} \label{cor:1}
Under  \ref{ass:1} and  \ref{ass:2new},  and  further assume  $m(n) \asymp n$.  
\begin{itemize}
    \item[(i)] When  $N_i \asymp C$ for all $i=1,2, \dots, n+m$, where $0 < C <\infty$ is a constant, and   $h_x  \asymp h_y \asymp n^{-1/5}$, we have
    \begin{align*}
	\left| \emph{\text{MED}}_n(\mathbf{Z}) - \emph{\text{MED}}(X,Y) \right| = O_p \left( \frac{1}{n^{2/5}} \right).
	\end{align*}
	
    \item[(ii)] When  $N_i \asymp n^{1/4}$ for all $i=1,2, \dots, n+m$ and    $h_x \asymp h_y \asymp n^{-1/4}$, we have 
    \begin{align*}
	\left| \emph{\text{MED}}_n(\mathbf{Z}) - \emph{\text{MED}}(X,Y) \right| = O_p \left( \frac{1}{\sqrt{n}} \right).
	\end{align*}
\end{itemize}
\end{corollary}

\subsection{Validity of the Permutation Test and Power Analysis}
We now justify the permutation based test for sparsely observed functional data. Under the null hypothesis and the  mild assumption that $\{N_i\}$ are i.i.d across subjects, the size of the test can be guaranteed by the fact that the distribution of the sample is invariant under permutation. For a rigorous argument, see Theorem 15.2.1 in \cite{lehmann2006testing}. Thus, the permutation test based on the test statistic (\ref{eq:MED}) produces a legitimate size of the test. The power analysis is much more challenging and will be presented below.


Let $\pi \in \mathbb{P}_{n+m}$ be a fixed permutation and $\widehat{G}_{\pi, I}(t_1, t_2), I=1,2,3$ be the estimated functions from algorithms \eqref{eq:xy} and \eqref{eq:xx} using permuted samples $\pi \cdot \mathbf{Z}$. The  conditions  on the decay rate of bandwidth parameters $h_x, h_y$ that ensure the convergence of $\widehat{G}_{\pi, I}$ for any fixed permutation $\pi$ are summarized below. 



\begin{assumption} \label{ass:power} 
 Suppose  $h_x:=h_x(n), h_y:=h_y(m) \rightarrow 0$ and
\begin{align*}
 & \sup\limits_{\pi \in \mathbb{P}_{n+m}} \log \left( \frac{n^2}{ \sum_{i=1}^n N_{\pi(i)}^{-1}} \right) \frac{\max\limits_{1 \leq i \leq n} N_{\pi(i)}^{-1}}{\min\{ h_x, h_y \} } \frac{\max\limits_{1 \leq i \leq n} N_{\pi(i)}^{-1}}{ \sum_{i=1}^n N_{\pi(i)}^{-1}/n} \frac{1}{ n}  \rightarrow 0, \\
  &\sup\limits_{\pi \in \mathbb{P}_{n+m}} \log \left( \frac{m^2}{ \sum_{i=n+1}^{n+m} N_{\pi (i)}^{-1}} \right) \frac{\max\limits_{n+1 \leq i \leq n+m} N_{\pi(i)}^{-1}}{\min\{h_x, h_y\} } \frac{\max\limits_{n+1 \leq i \leq n+m} N_{\pi(i)}^{-1}}{ \sum_{i=n+1}^{n+m} N_{\pi(i)}^{-1}/m} \frac{1}{m}  \rightarrow 0.
\end{align*}
\end{assumption}
Let $\Pi$ be a random permutation uniformly sampled from $\mathbb{P}_{n+m}$. If the sample is randomly shuffled, it holds that $Z_{\Pi(i)} =^d Z_{\Pi(j)}$ and $\text{MED}(Z_{\Pi(i)} , Z_{\Pi(j)} ) = 0$ for any $i,j=1,2, \dots, n+m$. For the sample estimate $\text{MED}_n(\Pi \cdot \mathbf{Z})$ based on the permuted sparse observations, we show that $\text{MED}_n(\Pi \cdot \mathbf{Z})$ converges to 0 in probability.

\begin{theorem} \label{thm:per}
Under Assumptions \ref{ass:1} and \ref{ass:power}, 
\begin{align*}
    \left|\emph{\text{MED}}_n( \Pi \cdot \mathbf Z) \right| =O_p \left(  \sup\limits_{\pi \in \mathbb{P}_{n+m}} \sqrt{ \frac{1}{n^2}  \sum_{i=1}^{n} \phi_{\pi(i)} }   + \sup\limits_{\pi \in \mathbb{P}_{n+m}} \sqrt{ \frac{1}{m^2}  \sum_{i=n+1}^{n+m} \phi_{\pi(i)} }  \right),
\end{align*}
where $ \Pi \sim \text{Uniform}(\mathbb P_{n+m})$ and is independent of the data.
\end{theorem}

For the original data that have not been permuted, $ \text{MED}_n(\mathbf{Z}) \rightarrow^p \text{MED}(X , Y )$, which is strictly positive under the alternative hypothesis. On the other hand, we know from Theorem \ref{thm:per} that the permuted statistics converges to 0 in probability. This suggests that,  under mild assumptions, the probability of rejecting the null approaches 1 as $n,m \rightarrow \infty$.
We make this idea rigorous in the following theorem.

\begin{theorem} \label{thm:power}
Under Assumption \ref{ass:1} and \ref{ass:power}, for any fixed $S > 1/ \alpha$ we have 
\begin{align*}
P_{H_A} \left( \widehat{p} \leq \alpha \right) \rightarrow 1.
\end{align*}
\end{theorem}


	


\begin{remark}
Since Assumption  3 implies Assumption 2, Theorem 1 holds under the assumption of Theorem \ref{thm:power} and it facilitates the proof of Theorem 3.
\end{remark}

\subsection{Handling  of Measurement Errors} \label{sec:correct}

With measurement errors present,  the actual observed data are: 
\begin{align*}
\widetilde{x}_{ij} &= X_i(T_{ij})+  e_{ij},  \; i=1,2, \dots, n, \; j = 1,2, \dots, N_i , \\
\widetilde{y}_{ij} &= Y_{i}(T_{ij})+e_{ij},  \; i = n+ 1, \dots, n+ m, \;  j = 1,2, \dots, N_i.
\end{align*}
where $ \{ e_{ij} : i=1,2 , \dots, n, j=1,2, \dots, N_i \} \sim^{i.i.d} e_1 $,  $ \{e_{ij} : i=n+1, \dots, n+m, j=1,2, \dots, N_i \} \sim^{i.i.d} e_2 $, and $e_1, e_2$ are mean 0 independent univariate random variables.  
Denote the combined noisy observations by  $ \widetilde{\mathbf{Z}}  = (\widetilde{\mathbf{z}}_1 ; \dots ; \widetilde{\mathbf{z}}_{n+m}) $, where
\begin{align*}
\widetilde{\mathbf{z}}_i= (\widetilde{z}_{i1}, \widetilde{z}_{i2}, \dots, \widetilde{z}_{i,N_i})^T = \left\lbrace
    \begin{array}{ll}
    (\widetilde{x}_{i1}, \widetilde{x}_{i2}, \dots , \widetilde{x}_{i,N_i})^T,   & \text{ if } 1 \leq i \leq n,  \\
    (\widetilde{y}_{i1}, \widetilde{y}_{i2}, \dots , \widetilde{y}_{i,N_i})^T,   & \text{ if } n+ 1 \leq i \leq n +m.
    \end{array} \right.
\end{align*}
The local linear smoothers described in Equations \eqref{eq:xy} and \eqref{eq:xx} are then applied with the input data $\{ z_{i_1j_1} \}$ and $\{ z_{i_2j_2} \}$ replaced, respectively, by $\{ \widetilde{z}_{i_1j_1} \}$ and $\{ \widetilde{z}_{i_2j_2} \}$. The resulting outputs  are denoted as $ \widehat{H}_{1}, \widehat{H}_{2}, \widehat{H}_{3}$, leading to the estimator  
\begin{align} \label{eq:medZtilde}
\text{MED}_n(\widetilde{\mathbf{Z}}) = \int 2 \widehat{H}_{1}(t,t) - \widehat{H}_{2}(t,t) - \widehat{H}_{3}(t,t) dt.
\end{align}
Correspondingly, the proposed test with contaminated data $\widetilde{\mathbf{Z}}$ is:
	$$
	\text{Reject } H_0,  \text{ if } \widetilde{p} \leq \alpha. 	
	$$
	where $ \widetilde{p} = \frac{1}{S} \big\{ 1 + \sum_{l=1}^{S-1}  \mathbb{I}_{ \{ \text{MED}_n( \Pi_l \cdot \widetilde{\mathbf{Z}}) \geq \text{MED}_n(  \widetilde{\mathbf{Z}})  \} } \big\}. $ To study the convergence of $\text{MED}_n(\widetilde{\mathbf{Z}})$, define the two-dimensional functions $ H_{1}(s,t) $, $H_{2}(s,t) $ and $ H_{3}(s,t) $  as
\begin{align} \label{tildeG}
\begin{split}
    & H_{1}(s,t) =  E[ |X(s) + e_{1} - Y(t) - e_2 | ] , \\
    &  H_{2}(s,t) = E[ | X(s) + e_1 - X'(t) - e_{1}' | ], \\
    & H_{3}(s,t) = E[ |Y(s) + e_2 - Y'(t) - e_{2}' | ],
\end{split}
\end{align}
where $e_{1}'$ and $e_{2}'$ are independent and identical copies of $e_1$ and $e_2$ respectively. The target of $\text{MED}_n(\widetilde{\mathbf{Z}})$ is shown  to be 
\begin{align}
\widetilde{\text{MED}}(X,Y) = \int 2 H_{1}(t,t) - H_{2}(t,t) - H_{3}(t,t) dt,
\end{align}

\begin{remark}
An unpleasant fact is that $\widetilde{\emph{\text{MED}}}(X, Y)$
 involves errors, which cannot be easily removed due to the presence of the absolute error function in (\ref{tildeG}).
The handling of measurement errors in both method and theory is thus very different here from  conventional approaches for functional data, where one does not deal with the   absolute function. The energy distance with $L^2$-norm in Equation \eqref{eq:2step} also has this issue. Thus, measurement errors would also be a challenge for the full homogeneity test even if we can approximate the $L^2$ norm in the energy distance well.
\end{remark}
To show the approximation error of $\text{MED}_n(\widetilde{\mathbf{Z}})$ we  need the following assumptions.
\begin{assumption} \label{ass:2}  
	\begin{itemize}
	    \item[D.1] $ E[e_1^2] < \infty  $ and $E[e_2^2] < \infty$.
		\item[D.2] $\{ X_{i_1}, Y_{i_2}, T_{ij}, e_{ij} : 1 \leq i_1 \leq n, n+1 \leq i_2 \leq n+m, 1 \leq i \leq n+m, 1 \leq j \leq N_i \}$ are mutually independent.
		\item[D.3] \label{ass:a4} The second order partial derivatives of $H_{1}, H_{2}, H_{3}$ are bounded on $[0,1]$.
	\end{itemize}
\end{assumption}
\begin{remark}
Using the notations $f_x(\cdot|t)$ and $f_y(\cdot|t)$ in Remark \ref{eq:rmk}, let the density functions of $e_1$ and $e_2$ be $\eta_{1}( \cdot )$ and $\eta_{2}(\cdot )$ respectively. Under the conditions of the Leibniz integral rule,
\begin{align*}
\frac{\partial^2}{\partial s \partial t} H_{1}(s,t) = \int |u-v + a-b| \frac{\partial}{\partial s} f_{x}(u|s) \frac{\partial}{ \partial t} f_{y}(v|t) \eta_{1}(a) \eta_{2}(b)dudvdadb,
\end{align*}
which admits bounded second-order partial derivatives if \eqref{eq:partial} holds. Similar conclusions can be drawn for $ H_2 $ and $H_3$. Therefore, Assumption D.3 is mild.
\end{remark}

\begin{corollary}
 \label{cor:2} Under Assumptions \ref{ass:1}, \ref{ass:2new} and \ref{ass:2}, we have
	\begin{align*}
	\left| \emph{\text{MED}}_n(\widetilde{\mathbf{Z}}) - \widetilde{\emph{\text{MED}}}(X,Y) \right| = O_p\left(h_x^2+ \sqrt{ \frac{1}{n^2}  \sum_{i=1}^{n} \phi_i }   + h_y^2 + \sqrt{ \frac{1}{m^2}  \sum_{i=n+1}^{n+m} \phi_i } \right).
	\end{align*}
\end{corollary}
 By the property of energy distance, it holds that $ \widetilde{\text{MED}}(X, Y) \geq 0 $ and $ \widetilde{\text{MED}} (X, Y ) = 0 \Leftrightarrow X(t) + e_1 =^d Y(t) + e_2$ for almost all $t \in [0,1]$. Then, we show that under the following assumptions, the condition that $ X(t) + e_1 =^d Y(t) + e_2 $ would imply the homogeneity of $X(t)$ and $Y(t)$.
 
\begin{assumption} \label{ass:3} Suppose that for any $t \in [0,1]$
\begin{itemize}
    \item[E.1] $X(t)$, $Y(t)$ are continuous random variables with density functions $f_{x}(\cdot|t)$, $f_{y}(\cdot|t)$ respectively.
    \item[E.2] $e_1 $, $e_2$ are i.i.d continuous random variables with characteristic function $\phi(\cdot)$ and the real zeros of $\phi(\cdot)$ has Lebesgue measure 0.
    \item[E.3] $ \{ X(t), Y(t), e_1, e_2 \} $ are mutually independent.
\end{itemize}
\end{assumption}

For common distributions, such as Gaussian and Cauchy, their characteristic functions are of exponential form and have no real zeros. Some other random variables, such as Exponential, Chi-square and Gamma, have characteristic functions of the form $\psi(t) = (1-it \theta)^{-k}$ with only a  finite number of real zeros. Since it is common to assume  Gaussian measurement errors, the restriction on the real zeros of the characteristic function in Assumption \ref{ass:3} (E.2) is  very mild.

\begin{theorem} \label{thm:2}
Under Assumption \ref{ass:3}, for any $t \in [0,1]$,
\begin{align*}
& X(t) + e_1 = ^d Y(t) + e_2 \Leftrightarrow X(t) =^d Y(t).
\end{align*}
\end{theorem}
Based on the above theorem, we have the important property that $ \widetilde{\text{MED}} (X,Y) = 0 \Leftrightarrow X(t)  =^d Y(t) $ for almost all $t \in [0,1]$. As discussed before, $\widetilde{\text{MED}}(X,Y)$ can be consistently estimated by $\text{MED}_n(\widetilde{\mathbf{Z}}) $ and the test can be conducted via permutations. Consequently,  the data contaminated with measurement errors can still be used to test marginal homogeneity as long as $e_1 = ^d e_2$. We make this statement rigorous below.  

\begin{corollary} \label{cor:powerconta}
Under Assumptions \ref{ass:1},  \ref{ass:power}, \ref{ass:2}, \ref{ass:3}. For any fixed $S > 1/ \alpha$,
\begin{align*}
P_{H_A} \left( \widetilde{p} \leq \alpha \right) \rightarrow 1.
\end{align*}
\end{corollary}

The circumstance of
Identically distributed errors among the two samples is a strong assumption that nevertheless can be satisfied in
 many real situations, for example,  when the curves $\{X_i\}$ and $\{ Y_j \}$ are measured by the same instrument. The PBC data in Section \ref{sec:realdata} underscore this phenomenon. 

When $ e_1 \neq ^d e_2 $, not all is lost and we show that some workarounds exist. In particular,
we propose an error-augmentation method that raises the noise of one sample  to the same level as that of the other sample. 


For instance, suppose that $e_1 \sim N(0, \sigma_{1}^2)$ and $ e_2 \sim N(0, \sigma_{2}^2)$. Then the variances $\sigma_{1}^2$ and $\sigma_{2}^2$ can be estimated consistently using the R package ``fdapace" \citep{fdapace} under both intensive  and sparse designs with  estimates $\widehat{\sigma}_{1}^2 $ and $\widehat{\sigma}_{2}^2$ respectively.  \cite{yao2005functional} showed that 
\begin{align*}
\left| \widehat{\sigma}_1 - \sigma_1 \right| = O_p\left( \frac{1}{\sqrt{n}} \left( \frac{1}{h_{G_x}^2} + \frac{1}{h_{V_x}} \right) \right), 
\end{align*}
where $ h_{G_x}$ and $ h_{V_x} $  are the bandwidth parameters for estimating the covariance function $\text{cov}(X(s), X(t))$ and the diagonal function $ \text{cov}(X(t), X(t)) + \sigma_1^2 $ respectively.  A different estimator for $\widehat{\sigma}_1$ that  has a better convergence rate is provided in \cite{linandwang2022}. 
An analogous result holds for $ \widehat{\sigma}_2 $. Then, by adding additional Gaussian white noise, we  obtain the error-augmented data $\{ \breve{x}_{ij} \}$, $\{ \breve{y}_{ij}\}$ as follows, 
\begin{align*} 
\left\lbrace 
\begin{array}{ll}
\left.
\begin{array}{ll}
 \breve{x}_{ij} = \widetilde{x}_{ij} + \epsilon_{ij}, & \text{ for } i=1,2, \dots, n, j=1, \dots, N_i,    \\
  \breve{y}_{ij} = \widetilde{y}_{ij}, & \text{ for } i=n+1, \dots, n+m, j=1, \dots, N_i,
\end{array} \right\rbrace
  & \text{ if }  \widehat{\sigma}_{1}^2 < \widehat{\sigma}_{2}^2, \\ [\bigskipamount]
  \left.
  \begin{array}{ll}
    \breve{x}_{ij} = \widetilde{x}_{ij},  & \text{ for } i=1,2, \dots, n, j=1, \dots, N_i , \\
    \breve{y}_{ij} = \widetilde{y}_{ij}+ \epsilon_{ij}, & \text{ for } i=n+1, \dots, n+m, j=1, \dots, N_i,
  \end{array} \right\rbrace
 & \text{ if }  \widehat{\sigma}_{1}^2 > \widehat{\sigma}_{2}^2,
\end{array} \right. 
\end{align*}
where $\{ \epsilon_{ij} \} \sim^{i.i.d} N(0,  |\widehat{\sigma}_{2}^2 - \widehat{\sigma}_{1}^2| )$.  With $\breve{\mathbf{Z}}$ being the combined error-augmented data, the proposed test is:
	$$
	\text{Reject } H_0 \text{ if } \breve{p} \leq \alpha,
	$$
where $ \breve{p} = \frac{1}{S} \big\{ 1 + \sum_{l=1}^{S-1}  \mathbb{I}_{ \{ \text{MED}_n( \Pi_l \cdot \breve{\mathbf{Z}}) \geq \text{MED}_n(  \breve{\mathbf{Z}})  \} } \big\}$.

The normal error assumption is common in practice. The  variance augmentation approach also works for any parametric family of distributions that is close under convolution (i.e., the sum of two independent distributions from this family is also a member of the family) and that has the property  that the first two moments of a distribution determine the distribution.

\section{Numerical Studies} \label{sec:sim}
In this section, we examine the proposed testing procedure for both synthetic and real data sets. 

\subsection{Performance on simulated data}
For simulations,  we set $\alpha = 0.05$ and perform 500 Monte Carlo replications with 200 permutations for each test. The following example is used to examine the size of our test.


\begin{example} \label{exp:1}
The stochastic processes $\{ X_i \}_{i=1}^n$ are i.i.d copies of $X$ and $\{ Y_{i} \}_{i=n+1}^{n+m}$ are i.i.d copies of $Y$,  where for $t\in [0, 1]$,
\begin{align*}
& X(t)  = \xi_1 \left( -\cos(2\pi t) \right) + \xi_2 \left( \sin(2 \pi t) \right) ,  \\
& Y(t) = \varsigma_1 \left( -\cos(2\pi t) \right) + \varsigma_2 \left( \sin(2 \pi t) \right), 
\end{align*}
and $  \xi_1, \xi_2, \varsigma_1, \varsigma_2, \sim^{i.i.d} N(0,1) $. These curves are observed at discrete time points
\begin{align*}
 & \widetilde{x}_{ij} = X_i(T_{ij}) + e_{ij}, i = 1,2, \dots, n, j=1,2, \dots, N_i, \\
 & \widetilde{y}_{ij} = Y_i(T_{ij}) + e_{ij}, i = n+1, \dots, n+m, j=1,2, \dots, N_i, 
\end{align*}
where $ \{ T_{ij} : i = 1, 2, \dots, n+m, j = 1,2, \dots, N_i \} \sim^{i.i.d} \text{Uniform}[0,1] $ and the measurement errors are Gaussian 
\begin{align*}
& \{ e_{ij}  : i = 1,2, \dots, n, j=1,2, \dots, N_i \} \sim^{i.i.d} N(0, \sigma_{1}^2), \\
& \{ e_{ij} : i=n+1, \dots, n+m, j=1,2, \dots, N_i \} \sim^{i.i.d} N(0, \sigma_{2}^2).
\end{align*}
\end{example}

\begin{table} 
\centering
\begin{tabular}{cccccc} \hline
$(n,m)$	 & $\sigma_{1}$ & $\sigma_{2}$ & $\text{MED}$   & FPCA \\ \hline
\multirow{3}{*}{$(100, 70)$}   & 0 & 0  & 0.05 & 0 \\
 & 0.2 & 0.2   & 0.04 & 0.004 \\
 & 0.05 & 0.25  & 0.058 & 0.006 \\ \hline
\multirow{3}{*}{$(150, 130)$}   & 0 & 0    & 0.054 & 0.002 \\
  & 0.2 & 0.2    & 0.046 & 0.002 \\
 & 0.05 & 0.25   & 0.056 & 0 \\ \hline
\end{tabular}	
\caption{Comparison of Type I errors  under the sparse design.}
\label{tab:size}
\end{table}

The quantities $\{N_i\}$ are used to control the sparsity level. Under sparse designs, $N_i$ are uniformly selected from $2\sim10$ for all $i=1,2, \dots, n+m$, and  $ \sigma_{1}^2 $ and $\sigma^2_{2}$ account for the magnitude of measurement errors. If $ \sigma_{1}^2 = \sigma^2_{2} =0 $, this corresponds to the case that there is no measurement error. If $ \sigma_{1}^2 \neq \sigma^2_{2} $, the error augmentation method described in Section \ref{sec:correct} is used. When applying the $\text{MED}$-based permutation test, the bandwidth parameters are selected as $h_x=h_y=0.2$.

Table \ref{tab:size} contains the size comparison results under this sparse design. As a baseline method for comparison, the FPCA approach is also included, where we first impute the principal scores and then apply the energy distance on the imputed scores. To be more specific, the first step is to reconstruct the principal scores using the R package ``fdapace" \citep{fdapace}. Then, a two sample tests for multivariate data is applied on the recovered scores. For this,  we choose the energy distance based procedure and conduct the hypothesis testing via the R package ``energy" \citep{eneryRpackage}. The FPCA approach has two drawbacks. First, the scores can not be estimated consistently; second, the infinite dimensional vector of scores has to be truncated for computational purposes, which  causes information loss. When the Gaussian errors have different variances, the same error-augmentation method is applied to the FPCA approach.  From Table \ref{tab:size}, we see that the $\text{MED}$-based methods have satisfactory size, while the FPCA based approach is undersized, likely due to the inherent inaccuracy in the imputed scores. 

\begin{example} \label{exp:size}
The stochastic processes $\{ X_i \}_{i=1}^n$ are i.i.d copies of $X$, $\{ Y_{i} \}_{i=n+1}^{n+m}$ are i.i.d copies of $Y$, where for $t\in [0,1]$, 
\begin{align*}
& X(t)  = \xi_1 \left( -\cos(2\pi t) \right) + \xi_2 \left( \sin(2 \pi t) \right) ,  \\
& Y(t) = \varsigma_1 t ^2 + \varsigma_2 \left( \sqrt{1- t^4} \right), 
\end{align*}
  $  \xi_1, \xi_2 \sim^{i.i.d} N(0,1), $ and  $\varsigma_1, \varsigma_2$ are independently sampled from the following mixture of Gaussian distributions, 
\begin{align*}
& P(\varsigma \leq a) = \frac{1}{2} P(N(\mu_{\varsigma},\sigma^2_{\varsigma}) \leq a) + \frac{1}{2 } P(  N(-\mu_{\varsigma},\sigma^2_{\varsigma}) \leq a ) \text{ for any } a \in \mathbb{R}.
\end{align*} 
The curves are observed at discrete time points
\begin{align*}
 & \widetilde{x}_{ij} = X_i(T_{ij}) + e_{ij}, i = 1,2, \dots, n, j=1,2, \dots, N_i, \\
 & \widetilde{y}_{ij} = Y_i(T_{ij}) + e_{ij}, i = n+1, \dots, n+m, j=1,2, \dots, N_i, 
\end{align*}
where $ \{ T_{ij} : i = 1, 2, \dots, n+m, j = 1,2, \dots, N_i \} \sim^{i.i.d} \text{Uniform} \ [0,1] $ and the measurement errors are Gaussian,  
\begin{align*}
& \{ e_{ij}  : i = 1,2, \dots, n, j=1,2, \dots, N_i \} \sim^{i.i.d} N(0, \sigma_{1}^2), \\
& \{ e_{ij} : i=n+1, \dots, n+m, j=1,2, \dots, N_i \} \sim^{i.i.d} N(0, \sigma_{2}^2).
\end{align*}
\end{example}

In the above example, $X(s)$ and $Y(t)$ both have zero mean functions. The variances of $\varsigma_1$ and $\varsigma_2$ are the same and equal to $ \mu_{\varsigma}^2 + \sigma^2_{\varsigma} $.  By selecting $\mu_{\varsigma}$ and $\sigma^2_{\varsigma}$ such that $ \mu_{\varsigma}^2 + \sigma^2_{\varsigma} =1 $, we have $\var(X(t)) = \var(Y(t))$.  Under this scenario, $X(s)$ and $Y(t)$ have the same marginal mean and variance, but different marginal distributions. In this example, we set $\mu_{\varsigma} = 0.98$ and $\sigma_{\varsigma} = 0.199$. 

The power comparison results are provided in Table \ref{tab:power}. The FPCA approach is not included for power analysis as we have already  shown  in Example \ref{exp:1} that it is undersized. Instead, we compare the power of $\text{MED}$-based tests under the sparse design with the result from dense regular data, where the sparse data are uniformly sampled. 
The same error-augmentation approach is applied for $\text{MED}$ with dense regular data when $ \sigma_1 \neq \sigma_2 $. From Table \ref{tab:power}, we can observe a moderate power drop from dense to sparse data. In addition, as the sample size increases, the power grows above 0.94 when $\sigma_1=\sigma_2$ and almost doubles when $\sigma_1 \neq \sigma_2$ under the sparse design.

\begin{table}
	\centering
	\begin{tabular}{ccccc} \hline
		$(n,m)$	 & $\sigma_{1}$ & $\sigma_{2}$ & $\text{MED}$ (dense) & $\text{MED}$ (sparse) \\ \hline
		\multirow{3}{*}{$(100, 70)$}  & 0 & 0    & 0.998 & 0.806 \\
		 & 0.2 & 0.2   & 0.866 & 0.53 \\
		 & 0.05 & 0.25  & 0.76 & 0.412 \\ \hline
		\multirow{3}{*}{$(150, 130)$}  & 0 & 0    & 1 & 0.998 \\
		  & 0.2 & 0.2   & 1 & 0.946 \\
		  & 0.05 & 0.25   & 0.998 & 0.786 \\ \hline 
	\end{tabular}	
	\caption{Power comparison.} 
	\label{tab:power}
\end{table}

\subsection{Applications to real  data} \label{sec:realdata}
In this subsection, we apply the proposed $\text{MED}$-based tests to two real data sets. \\
\noindent \textbf{PBC Data}: \ The first data set  is the primary biliary cirrhosis (PBC) data from Mayo Clinic \citep{fleming2011counting}. This data set is from a clinical trial studying primary biliary cirrhosis of the liver. There were 312 patients assigned to either the treatment or control group. The drug D-penicillamine is given to the treatment group. Here, we are interested in testing the equality of the marginal distributions of Prothrombin time, which is a blood test that measures how long it takes blood to clot.  
The trajectories of Prothrombin time for different subjects are plotted in Figure \ref{fig:straw}, there are on average 6 measurements per subjects.  
For our tests, the bandwidth is set to be 2. 
Here, the equal distribution assumption for errors seems to work (the estimated variances for treatment and control group are 0.96 and 1.009 respectively). 
By using 200 permutations, the $p$-value of the $\text{MED}$-based test is 0.54, which means that there is not enough evidence to conclude that the maginal distributions of Prothrombin time are different between the two groups. This conclusion matches with existing knowledge that D-penicillamine is ineffective to treat primary biliary cirrhosis of liver. \\


\noindent \textbf{Strawberry Data}: \ In the  food industry, there is a continuing interest in distinguishing the pure fruit purees from the adulterated ones \citep{straw}. One practical way to detect adulteration is by looking at the spectra of the fruit purees. Here, we are interesting in testing the marginal distribution between the spectra of strawberry purees (authentic samples) and non-strawberry purees (adulterated strawberries and other fruits). The strawberry data can be downloaded from the UCR Time Series Classification Archive \citep{UCRArchive2018} (\url{https://www.cs.ucr.edu/~eamonn/time_series_data_2018/}). The single-beam spectra of the purees were normalized to back-ground spectra of water and then transformed into absorbance units. The spectral range was truncated to 899-1802 $\text{cm}^{-1}$ (235 data points).  The two samples of spectra are plotted in Figure \ref{fig:straw} and more information about this data set can be found at \cite{straw}. The estimated variances of the measurement errors are 0.000279 and 0.00031 for the two samples, which indicates that there are practically no measurement errors. To check the performance of our method, we analyze the data using all 235 measurements as well as a sparse subsamples  that contain  $2$ to  $10$ observations per subject. The R package ``energy" is applied for the complete data. Both tests are conducted with 200 permutations and have $p$-value 0.005.  
Thus, we have strong evidence to conclude that the marginal distributions between the spectra of strawberry and non-strawberry purees are significantly different and our test produced similar results regardless of the sampling plan.

\begin{figure}
\centering
\begin{subfigure}{.5\textwidth}
  \centering
    \caption{PBC}
  \label{fig:sub1}
  \includegraphics[scale=0.5]{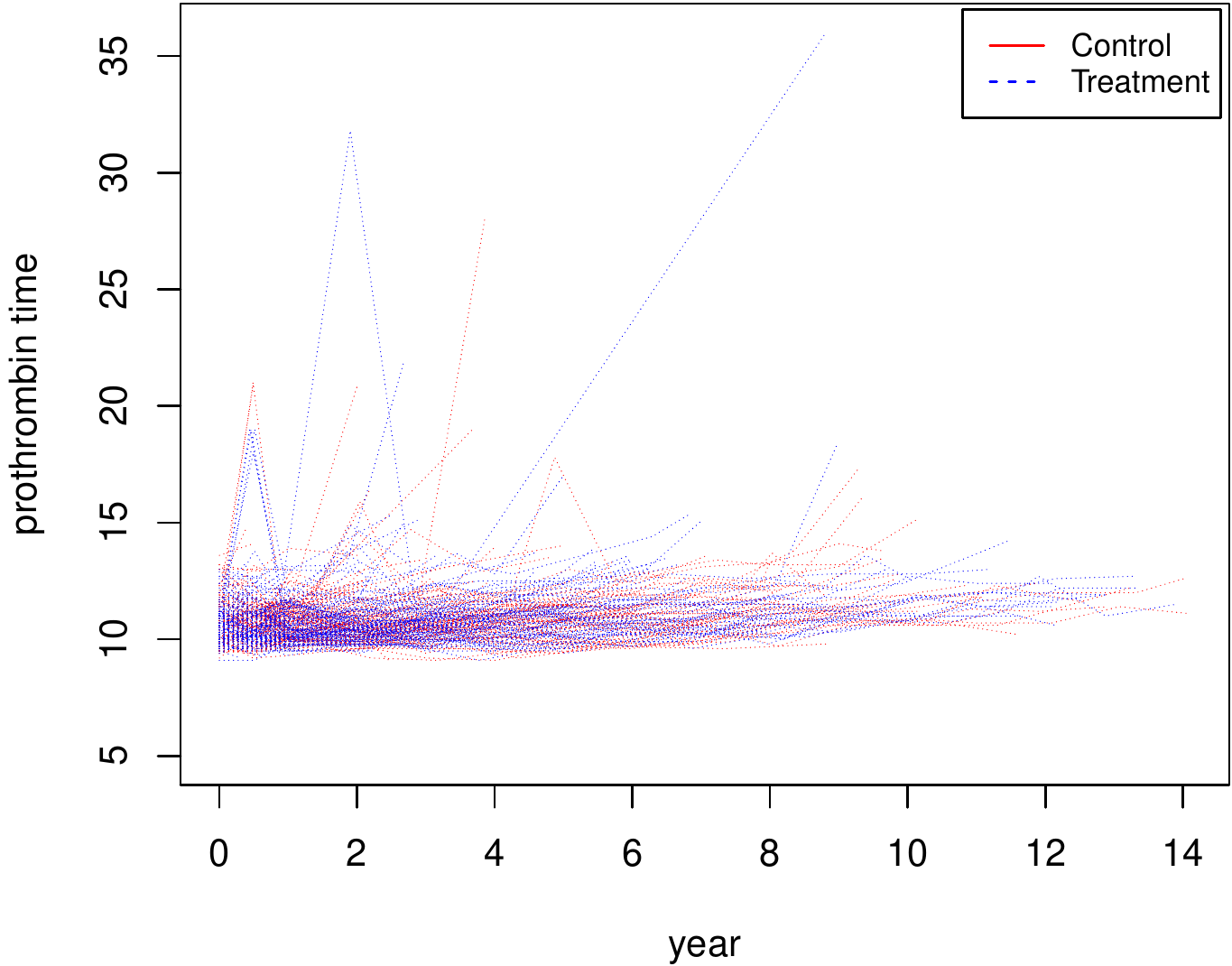}
\end{subfigure}%
\begin{subfigure}{.5\textwidth}
  \centering
    \caption{Strawberry}
  \label{fig:sub2}
  \includegraphics[scale=0.5]{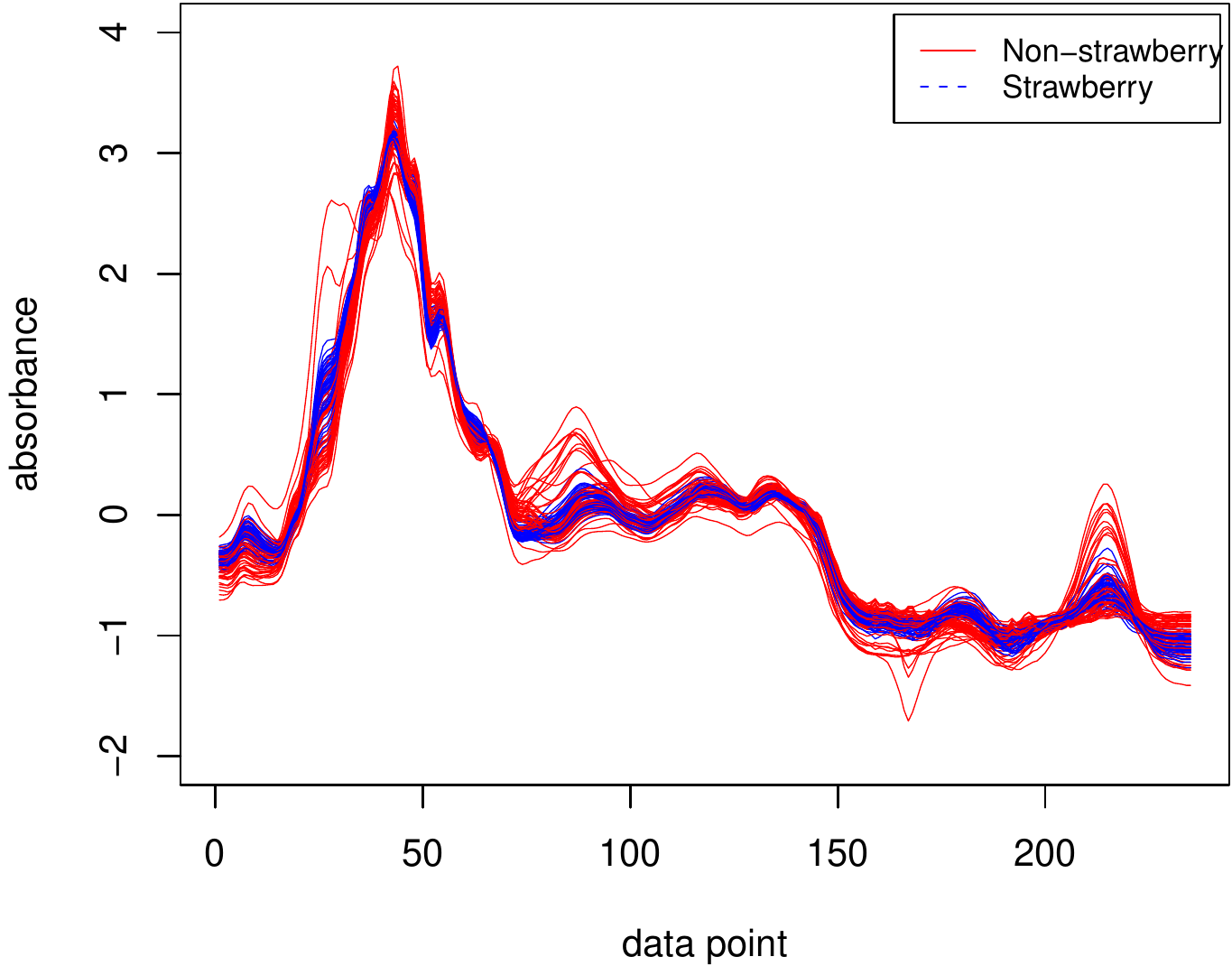}
\end{subfigure}
\caption{Trajectories of real data.}
\label{fig:straw}
\end{figure}

\section{Conclusion} \label{sec:con}
The literature on testing homogeneity for functional data is scarce probably because most approaches rely on intense measurement schedules and the hope that measurement errors can be addressed by presmoothing the data.  Since reconstruction of noise-free functional data is not feasible for sparsely observed functional data, a test of homogeneity is infeasible. In this work, we show what is feasible for sparse functional data, a.k.a. longitudinal data, and propose a test of marginal homogeneity that 
adapts to the sampling plan and provides the corresponding convergence rate.  Our test is based on  Energy distance with a focus on testing the marginal homogeneity. 
To the best of our knowledge, this is the only nonparametric test with theoretical guarantees under sparse designs, which are ubiquitous. 

There are several twists in our approach, including the handling of asynchronized longitudinal data and the unconventional way that measurement errors affect the method and theory.  
The asynchronization of the data can be overcome completely as we demonstrated in Section 2.1, but the handling of measurement errors requires some compromise when the distributions of the measurement errors are different for the two samples. This is the price one pays for lack of data, and is not due to the use of the $L_1$ norm associated with testing the marginal homogeneity, as an $L^2$ norm for testing full homogeneity would also face the same challenge with measurement errors unless a presmoothing step has been employed to eliminate the measurement errors.   As we mentioned in Section 1 this would require a super intensive sampling plan well beyond the usual requirement for dense or ultra dense functional data \citep{zhang2016sparse}.  
While the new approach may involve error-augmentation, numerical results show that the efficiency loss is minimal. Moreover, such an augmentation strategy is not uncommon. For instance, an error augmentation method has also been adopted in the SIMEX approach \citep{doi:10.1080/01621459.1994.10476871} to deal with measurement errors for vector data.

While testing marginal homogeneity has its own merits and advantages over a full-fledged test of homogeneity, our intention is not to particularly endorse it.   Rather, we  point  out what is feasible and infeasible for sparsely or intensively measured functional data and develop theoretical support for the proposed test. To the best of our knowledge, we are the first to provide the convergence rate for the permuted statistics for sparse functional data. This proof and the proof of consistency for the proposed permutation test is non-conventional and different from the the multivariate/high-dimensional case.


\bibliographystyle{rss}
\bibliography{JRSSB}
\end{document}